\documentclass[10pt]{article}

\usepackage{fullpage}
\usepackage{citesort}
\usepackage{times}
\usepackage{epsfig}
\usepackage{color}
\usepackage{amssymb}
\usepackage{jsv1}

\newcommand{\eps}{\epsilon}

\title{Optimal Tracking of Distributed Heavy Hitters and Quantiles}

\author{Ke Yi \and Qin Zhang}

\date{Department of Computer Science and Engineering \\
Hong Kong University of Science and Technology \\
$\{$yike, qinzhang$\}$cse.ust.hk}

\begin{document}

\maketitle

\begin{abstract}
  We consider the the problem of tracking heavy hitters and quantiles in
  the distributed streaming model.  The heavy hitters and quantiles are two
  important statistics for characterizing a data distribution.  Let $A$ be
  a multiset of elements, drawn from the universe $U=\{1,\dots,u\}$.  For a
  given $0 \le \phi \le 1$, the $\phi$-heavy hitters are those elements of
  $A$ whose frequency in $A$ is at least $\phi |A|$; the $\phi$-quantile of
  $A$ is an element $x$ of $U$ such that at most $\phi|A|$ elements of $A$
  are smaller than $A$ and at most $(1-\phi)|A|$ elements of $A$ are
  greater than $x$.  Suppose the elements of $A$ are received at $k$ remote
  {\em sites} over time, and each of the sites has a two-way communication
  channel to a designated {\em coordinator}, whose goal is to track the set
  of $\phi$-heavy hitters and the $\phi$-quantile of $A$ approximately at
  all times with minimum communication.  We give tracking algorithms with
  worst-case communication cost $O(k/\eps \cdot \log n)$ for both problems,
  where $n$ is the total number of items in $A$, and $\eps$ is the
  approximation error.  This substantially improves upon the previous known
  algorithms.  We also give matching lower bounds on the communication
  costs for both problems, showing that our algorithms are optimal.  We
  also consider a more general version of the problem where we
  simultaneously track the $\phi$-quantiles for all $0 \le \phi \le 1$.
\end{abstract}


\section{Introduction}
Data streams have been studied in both the database and theory communities
for more than a decade \cite{alon99,babcock02:_model_}.  In this model,
data items arrive in an online fashion, and the goal is to maintain some
function $f$ over all the items that have already arrived using small
space.  A lot of $f$'s have been considered under the streaming model.  The
theory community have studied various frequency moments
\cite{alon99,indyk05:_optim_,woodruff04:_optim_}, geometric problems
\cite{Indyk:04,suri04:_range,agarwal07:_space_}, and some graph problems
\cite{feigenbaum05:_graph_,bar-yossef02:_reduc}.  While the database
community have mostly focused on maintaining the frequent items (a.k.a.\
{\em heavy hitters})
\cite{cormode08:_findin,cormode03,karp03,metwally06,manku02:_approx} and
quantiles \cite{gilbert02:_how,cormode06:_space,greenwald01:_space}, two
very important statistics for characterizing a data distribution.  Since we
cannot afford to store all the items, we can only maintain an approximate
$f$ (except for some trivial $f$'s), and all the results in the streaming
model are expressed as a tradeoff between the approximation error $\eps$
and the space used by the algorithm.  After a long and somehow disorganized
line of research, the heavy hitter problem is now completely understood
with both space upper and lower bounds determined at $\Theta(1/\eps)$;
please see the recent paper by Cormode and Hadjieleftheriou
\cite{cormode08:_findin} for a comprehensive comparison of the existing
algorithms for this problem, both theoretically and empirically.  For
maintaining quantiles, the best upper bound is due to a sketch structure by
Greenwald and Khanna \cite{greenwald01:_space}, using space $O(1/\eps\cdot
\log(\eps n))$ where $n$ is the number of items in the stream.  This is
conjectured to be optimal but not yet proved.

Recent years have witnessed an increasing popularity of another model more
general than the streaming model, where multiple streams are considered.
In this model, multiple streams are received at multiple distributed {\em
  sites}, and again we would like to continuously track some function $f$
over the union of all the items that have arrived across all the sites.
Here the most important measure of complexity is the total communication
cost incurred during the entire tracking period.  This model, which is
either referred to as the {\em distributed streaming model} or the {\em
  continuous communication model}, is a natural combination of the
classical communication model \cite{yao79} and the data stream model.
Recall that the communication model studies the problem of computing some
function $f$ over distributed data using minimum communication.  The
data is predetermined and stored at a number of sites, which communicate
with a central coordinator, and the goal is to do a {\em one-time}
computation of the function $f$.  Thus the distributed streaming model is
more general as we need to maintain $f$ continuously over time as items
arrive in a distributed fashion.

The rising interest on the distributed streaming model is mainly due to its
many applications in distributed databases, wireless sensor networks, and
network monitoring.  As a result, it has attracted a lot of attention
lately in the database community, resulting in a flurry of research in this
area
\cite{Cormode:Muthukrishnan:Zhuang:07,cormode06:what,keralapura06,Cormode:Garofalakis:Muthukrishnan:Rastogi:05,cormode05:sketch,Olston:Jiang:Widom:03,Babcock:Olston:03,fuller07:_fids,deshpande04:_model,manjhi05:_findin,olston05:_effic,sharfman08:_shape}.
However, nearly all works in this area are heuristic and empirical
in nature, with a few exceptions to be mentioned shortly.  For many
fundamental problems in this model, our theoretical understandings are
still premature.  This is to be contrasted with the standard streaming
model, where theory and practice nicely blend, and in fact many of the most
practically efficient solutions are the direct products of our theoretical
findings.  In this paper, we take an important step towards an analytical
study of the distributed streaming model, by considering the worst-case
communication complexity of tracking heavy hitters and quantiles, arguably
two of the most fundamental problems on data streams.

\paragraph{The distributed streaming model.}
We now formally define the distributed streaming model, which is the same
as in most works in this area.  Let $A=(a_1, \dots, a_n)$ be a sequence of
items, where each item is drawn from the universe $U=\{1,\dots,u\}$.  The
sequence $A$ is observed in order by $k\ge 2$ remote {\em sites} $S_1,
\dots, S_k$ collectively, i.e., item $a_i$ is observed by exactly one of
the sites at time instance $t_i$, where $t_1 < t_2 < \cdots < t_n$.  Let
$A(t)$ be the multiset of items that have arrived up until time $t$ from
all sites.  Then the general goal is to continuously track $f(A(t))$ for
some function $f$ at all times $t$ with minimum total communication among
the sites.  Note that in the classical communication model, the goal is to
just compute $f(A(+\infty))$; in the data stream model, the goal is to
track $f(A(t))$ for all $t$ but there is only one site ($k=1$), and we are
interested in the space complexity of the tracking algorithm, not
communication.  Thus, the distributed streaming model is a natural
combination of the two, but is also significantly different from either.

We define the manner of communication more precisely as follows.  There is
a distinguished {\em coordinator} $C$, who will maintain (an approximate)
$f(t)$ at all times.  There is a two-way communication channel between the
coordinator and each of the $k$ sites, but there is no direct communication
between any two sites (but up to a factor of 2, this is not a restriction).
Suppose site $S_j$ receives the item $a_i$ at time $t_i$.  Based on its
local status, $S_j$ may choose to send a message to $C$, which in turn may
trigger iterative communication with other sites.  We assume that
communication is instant.  When all communication finishes, all the sites
who have been involved may have new statuses, getting ready for the next
item $a_{i+1}$ to arrive.  We will measure the communication cost in terms
of words, and assume that each word consists of $\Theta(\log u) =
\Theta(\log n)$ bits.  Finally we assume that $n$ is sufficiently large
(compared with $k$ and $1/\eps$); if $n$ is too small, a naive solution
that transmits every arrival to the coordinator would be the best.

In this paper we will focus on the communication cost (or simply the {\em
  cost}).  Nevertheless, all the algorithms proposed in this paper can be
implemented both space- and time-efficiently.

\paragraph{Heavy hitters and quantiles.}
By taking different $f$'s, we arrive at different continuous tracking
problems.  The notion of $\eps$-approximation also differs for different
functions.  We adopt the following agreed definitions in the literature.
In the sequel, we abbreviate $A(t)$ as $A$ when there is no confusion.

For any $x\in U$, let $m_x(A)$ be the number of occurrences of $x$ in $A$.
For some user specified $0 \le \phi\le 1$, the set of {\em $\phi$-heavy
  hitters} of $A$ is $\mathcal{H}_\phi(A) = \{ x \mid m_x(A) \ge \phi
|A|\}$, where $|A|$ denotes the total number of items in $A$.  If an
$\eps$-approximation is allowed, then the returned set of heavy hitters
must contain $\mathcal{H}_\phi(A)$ and cannot include any $x$ such that
$m_x(A)<(\phi-\eps)|A|$.  If $(\phi - \eps)|A| \le m_x(A) < \phi|A|$, then
$x$ may or may not be reported.  In the {\em heavy hitter tracking}
problem, the coordinator should always maintain an approximate
$\mathcal{H}_\phi(A)$ at all times for a given $\phi$.

For any $0 \le \phi \le 1$, the {\em $\phi$-quantile} of $A$ is some $x \in
U$ such that at most $\phi |A|$ items of $A$ are smaller than $x$ and at
most $(1-\phi)|A|$ items of $A$ are greater than $x$.  The quantiles are
also called {\em order statistics} in the statistics literature.  In
particular, the $\frac{1}{2}$-quantile is also known as the {\em median} of
$A$.  If an $\eps$-approximation is allowed, we can return any
$\phi'$-quantile of $A$ such that $\phi - \eps \le \phi' \le \phi + \eps$.
In the {\em $\phi$-quantile tracking} problem, the coordinator needs to
keep an $\eps$-approximate $\phi$-quantile of $A$ at all times for a given
$\phi$.  We also consider a more general version of the problem, where we
would like to keep track of all the quantiles approximately.  More
precisely, here the ``function'' $f$ is a data structure from which an
$\eps$-approximate $\phi$-quantile for any $\phi$ can be extracted.  Note
that such a structure is equivalent to an (approximate) equal-height
histogram, which characterizes the entire distribution.

In particular, from an {\em all-quantile} structure, we can easily obtain
the $(2\eps)$-approximate $\phi$-heavy hitters for any $\phi$, as observed
in \cite{Cormode:Garofalakis:Muthukrishnan:Rastogi:05}.  Therefore, the
all-quantile tracking problem is more general than either the
$\phi$-heavy hitter tracking problem or the $\phi$-quantile tracking
problem.  In the rest of the paper, we omit the word ``approximate'' when
referring to heavy hitters and quantiles when the context is clear.

\paragraph{Previous works.}
Traditionally, query answering in distributed databases follows a ``poll''
based approach, that is, the coordinator collects information from the
sites to answer a query posed by the user using minimum communication.
Such a paradigm falls into the realm of the classical multi-party
communication theory.  These queries are also referred to as {\em one-shot}
queries in the literature.  As long-standing queries that need to be
answered continuously become common in many modern applications such as
sensor network monitoring, network anomaly detection, publish-subscribe
systems, etc., periodically polling all the sites is neither efficient nor
effective (i.e., long latency).  Thus, the trend is moving towards a
``push'' based approach \cite{Jain:Hellerstein:Ratnasamy:Wetherall:04}, in
which the sites actively participate in the tracking process.  In this
framework, each site maintains some local conditions, and will not initiate
communication unless one of the conditions is triggered.  Such an approach
often leads to much reduced communication overhead compared with the
``poll'' based approach, since the system will react only when
``interesting'' things are happening.  This is the main motivation that has
led to the distributed streaming model described above.

Various $f$'s have been considered under this framework.  The simplest case
$f(A) = |A|$ just counts the total number of items received so far across
all the sites.  This problem can be easily solved with $O(k/\eps \cdot \log
n)$ communication where each site simply reports to the coordinator
whenever its local count increases by a $1+\eps$ factor
\cite{keralapura06}.  The other important single-valued statistics are the
frequency moments: $F_p(A) = \sum_x (m_x(A))^p$.  $F_0$ is the number of
distinct items, and can be tracked with cost $O(k/\eps^2 \cdot \log n
\log\frac{n}{\delta})$ \cite{graham08}; $F_2$ is the self-join size and can
be tracked with cost $O((k^2/\eps^2 + k^{3/2}/\eps^4) \log n
\log\frac{kn}{\eps\delta})$ \cite{graham08}.  Some heuristic approaches
based on predicting future arrivals of items have been proposed in
\cite{cormode06:what,cormode05:sketch}.

Single-valued statistics have very limited expressive power, so
multi-valued statistics are often necessary to better capture the
distribution of data.  The most important ones include the heavy hitters
and quantiles, and they have also been studied under the distributed
streaming framework.  Babcock and Olston \cite{Babcock:Olston:03} designed
some heuristics for the top-$k$ monitoring problem, where the goal is to
track the $k$ most frequent items (whose frequency may not be larger than
$\phi |A|$).  Their techniques can be adapted to tracking the heavy hitters
\cite{fuller07:_fids}, but the approach remains heuristic in nature.
Manjhi et al.~\cite{manjhi05:_findin} also studied the heavy hitter
tracking problem, but their communication model and the goal are different:
They organize the sites in a tree structure and the goal is to minimize the
communication only at the root node.  The all-quantile tracking problem has
been studied by Cormode et
al.~\cite{Cormode:Garofalakis:Muthukrishnan:Rastogi:05}, who gave an
algorithm with cost $O(k/\eps^2 \cdot \log n)$.  As commented earlier, this
also implies a heavy hitter tracking algorithm with the same cost.  This
remains the best communication upper bound for both problems to date.  No
lower bound is known.

\paragraph{Our results.}
Our main results in this paper are the matching upper and lower bounds on
the communication cost for deterministic algorithms for both the heavy
hitter tracking problem and the quantile tracking problem.  Specifically,
we show that for any $\phi$, both the $\phi$-heavy hitters
(Section~\ref{sec:track-heavy-hitt}) and the $\phi$-quantile
(Section~\ref{sec:track-median}) can be tracked with total communication
cost $O(k/\eps \cdot \log n)$.  This improves upon the previous result of
\cite{Cormode:Garofalakis:Muthukrishnan:Rastogi:05} by a $\Theta(1/\eps)$
factor.  We also give matching lower bounds for both problems, showing that
our tracking protocols are optimal in terms of communication.  Note that in
the classical communication model, we can easily do a one-shot computation
of the $\phi$-heavy hitters and the $\phi$-quantile easily with cost
$O(k/\eps)$, as observed in
\cite{Cormode:Garofalakis:Muthukrishnan:Rastogi:05}.  Interestingly, our
results show that requiring the heavy hitters and quantiles to be tracked
at all times indeed increases the communication complexity, but only by a
$\Theta(\log n)$ factor.  In Section~\ref{sec:tracking-quantiles}, we give
an algorithm that tracks all quantiles with cost $O(k/\eps \cdot
\log^2\frac{1}{\eps}\log n)$.  Because this problem is more difficult than
the single-quantile problem, it has the same lower bound of $\Omega(k/\eps
\cdot \log n)$ as the latter.  Thus, our all-quantile tracking algorithm is
also optimal up to a $\Theta(\polylog\frac{1}{\eps})$ factor.

\section{Tracking the Heavy Hitters}
\label{sec:track-heavy-hitt}


\subsection{The upper bound}
\label{sec:upper-bound}
\paragraph{The algorithm.}
Let $m$ be the current size of $A$. First, the coordinator $C$
always maintains $C.m$, an $\eps$-approximation of $m$.  This can
be achieved by letting each site send its local count every time
it has increased by a certain amount (to be specified shortly).
Each site $S_j$ maintains the exact frequency of each $x \in U$ at
site $S_j$, denoted $m_{x,j}$, at all times. The overall frequency
of $x$ is $m_x = \sum_j m_{x,j}$. Of course, we cannot afford to
keep track of $m_x$ exactly.  Instead, the coordinator $C$ maintains
an underestimate $C.m_{x,j}$ of $m_{x,j}$, and sets $C.m_x =
\sum_j C.m_{x,j}$ as an estimate of $m_x$.  $S_j$ will send its
local increment of $m_{x,j}$ to $C$, hence updating $C.m_{x,j}$,
from time to time following certain rules to be specified shortly.
In addition, each site $S_j$ maintains ${S_j}.m$, an estimate of
$m$, a counter ${S_j}.\Delta(m)$, denoting the increment of
$S_j.m$ since its last communication to $C$ about $S_j.m$, as well
as a counter ${S_j}.\Delta(m_x)$ for each $x$, denoting the
increment of ${S_j}.m_x$ since its last communication to $C$ about
$m_{x,j}$.

We can assume that the system starts with $m = k/\eps$ items;
before that we could simply send each item to the coordinator.  So
when the algorithm initiates, all the estimates are exact. We
initialize ${S_j}.\Delta(m)$ and ${S_j}.\Delta(m_x)$ for all $x$
to be 0. The protocols of tracking the $\phi$-heavy hitters are as
follows.
\begin{enumerate}
\item {\em Each site $S_j$:} When a new item of $x$ arrives,
${S_j}.\Delta(m)$ and ${S_j}.\Delta(m_x)$ are incremented by 1.
When ${S_j}.\Delta(m)$ (resp.\ ${S_j}.\Delta(m_x)$) reaches $(\eps
\cdot {S_j}.m)/3k$, site $S_j$ sends a message $(all, (\eps \cdot
{S_j}.m)/3k)$ (resp.\ $(x,(\eps \cdot {S_j}.m)/3k)$) to the
coordinator, and resets ${S_j}.\Delta(m)$ (resp.\
${S_j}.\Delta(m_x)$) to 0.

\item {\em Coordinator:} When $C$ has received a message $(all,
(\eps \cdot {S_j}.m)/3k)$ or $(x,(\eps \cdot {S_j}.m)/3k)$, it
updates $C.m$ to $C.m + (\eps \cdot {S_j}.m)/3k$ or $C.m_x$ to
$C.m_x + (\eps \cdot {S_j}.m)/3k$, respectively. Once $C$ has
received $k$ signals in the forms of $(all, (\eps \cdot
{S_j}.m)/3k)$, it collects the local counts from each site to compute
the exact value of $m$, sets $C.m = m$, and then broadcasts $C.m$
to all sites.  Then each site $S_j$ updates its ${S_j}.m$ to $m$.
After getting a new $S_j.m$, $S_j$ also resets $S_j.\Delta(m)$ to
0.
\end{enumerate}

Finally, at any time, the coordinator $C$ declares an item $x$ to be a
$\phi$-heavy hitter if and only if
\begin{equation}
\label{eq:hh}
\frac{C.m_x}{C.m} \ge \phi + \frac{\eps}{2}.
\end{equation}

\paragraph{Correctness.}
To prove correctness we first establish the following invariants
maintained by the algorithm.
\begin{equation}
\label{eq:bound_C.m_x} m_x - \frac{\eps m}{3} + k \le C.{m_x} \le
m_x,
\end{equation}
\begin{equation}
\label{eq:bound_C.m} m - \frac{\eps m}{3} + k \le C.m \le m.
\end{equation}

The second inequalities of both (\ref{eq:bound_C.m}) and
(\ref{eq:bound_C.m_x}) are obvious. The first inequality of
(\ref{eq:bound_C.m_x}) is valid since once a site $S_j$ gets
$(\eps \cdot {S_j}.m)/3k$ items of $x$, it sends a message to the
coordinator and the coordinator updates $C.m_x$ accordingly. Thus
the maximum error of $C.m$ in the coordinator is at most
$\sum_{j=1}^k (\frac{\eps \cdot {S_j}.m}{3k}- 1) \le \frac{\eps
m}{3} - k$. The first inequality of (\ref{eq:bound_C.m}) follows
from a similar reason. Combining (\ref{eq:bound_C.m_x}) and
(\ref{eq:bound_C.m}), we have
$$\frac{m_x}{m} - \frac{\eps}{3} < \frac{C.m_x}{C.m} <
\frac{m_x}{m}\cdot\frac{1}{1-\eps/3} < \frac{m_x}{m} + \frac{\eps}{2},$$
which guarantees that the approximate ratio $\frac{C.m_x}{C.m}$ is within
$\eps/2$ of $\frac{m_x}{m}$, thus classifying an item using (\ref{eq:hh})
will not generate any false positives or false negatives.

\paragraph{Analysis of communication complexity.}
We divide the whole tracking period into rounds. A round start
from the time when the coordinator finishes a broadcast of $C.m$ to the
time when it initiates the next broadcast. Since the coordinator
initiates a broadcast after $C.m$ is increased by a factor of
$1+\sum_{i=1}^{k}(\eps/3k) = 1+\eps/3$, the number of rounds is
bounded by
$$\log_{1+\eps/3} n = O\left(\frac{\log n}{\eps}\right).$$

In each round, the number of messages in the form of $(all, (\eps \cdot
{S_j}.m)/3k)$ sent by all the sites is $k$ by the definition of our
protocol. Since there are $O(\log n/ \eps)$ rounds in total, the number of
messages in the form of $(all, (\eps \cdot {S_j}.m)/3k)$ can be bounded by
$O(k/\eps \cdot \log n)$. On the other hand, it is easy to see that total
number of messages of the form $(x, (\eps \cdot S_j.m)/3k)$ is no more than
the total number of messages of the form $(all, (\eps \cdot S_j.m)/3k)$.
Therefore, the total cost of the whole system is bounded by
$O(k/\eps \cdot \log n)$.

\begin{theorem}
For any $\eps \le \phi \le 1$, there is a deterministic algorithm that
  continuously tracks the $\phi$-heavy hitters and incurs a total
  communication cost of 
  $O(k/\eps\cdot \log n)$.
\end{theorem}

\paragraph{Implementing with small space.}
In the algorithm described above, we have assumed that each site maintains
all of its local frequencies $S_j.m_x$ exactly.  In fact, it is not
difficult to see that our algorithm still works if we replace these exact
frequencies with a heavy hitter sketch, such as the {\em space-saving}
sketch \cite{metwally06}, that maintains the local $\eps'$-approximate
frequencies for all items for some $\eps' = \Theta(\eps)$.  More precisely,
such a sketch gives us an approximate $S_j.m_x$ for any $x\in U$ with
absolute error at most $\eps' |S_j|$, where $|S_j|$ denotes the current
number of items received at $S_j$ so far.  We need to adjust some of the
constants above, but this does not affect our asymptotic results.  By using
such a sketch at each site, our tracking algorithm can be implemented in
$O(1/\eps)$ space per site and amortized $O(1)$ time per item.

\subsection{The lower bound} \label{sec:lower-bound-heavy-hitt}

To give a lower bound on the total communication cost that any
deterministic tracking algorithm must take, we first consider the number of
changes that the set of heavy hitters could experience, where a {\em
  change} is defined to be the transition of the frequency of an item from
above $\phi|A|$ to below $(\phi-\eps)|A|$, or the other way round.  Then we
show that to correctly detect each change, the system must exchange at
least a certain amount of messages. The following lemma could be
established by construction.

\begin{lemma}
\label{lem:numberOfOutputUpdates} For any $\phi>3\eps$, there is a
sequence of item arrivals such that the set of heavy hitters in
the whole tracking period will have $\Omega(\log n / \epsilon)$
changes.
\end{lemma}

\begin{proof}
  Set $\eps'=2\eps$. We construct two groups of $l= 1/(2\phi - \epsilon')$
  items each: $\mathcal{S}_0=\{t_1, t_2, \ldots, t_l\}$ and
  $\mathcal{S}_1=\{t_{l+1}, t_{l+2}, \ldots, t_{2l}\}$. Since we only care
  about the total number of changes of the set of heavy hitters during the
  whole tracking period, we temporarily treat the whole system as one big
  site and items come one by one. We will construct an input sequence under
  which the set of heavy hitters will undergo $\Omega(\log n /
  \epsilon)$ changes.

We still divide the whole tracking period to several rounds, and
let $m_i$ denote the total number of items when round $i$ starts.
The following invariant will be maintained throughout the
construction:
\begin{enumerate}
\item[] Let $b = i \bmod 2$.  When round $i$ starts, all items $t
\in \mathcal{S}_{b}$ have frequency $\phi m_i$, and all items $t
\in \mathcal{S}_{1-b}$ have frequency $(\phi-\epsilon') m_i$.
\end{enumerate}
It can be verified that the total frequency of all items is indeed
$m_i$. Note that from the start of round $i$ to the end of round
$i$, all the non-heavy hitters become heavy hitters, and all the
heavy hitters become non-heavy hitters. In what follows we only
care about the changes of the former type, which lower bounds
the number of changes. To maintain the
invariant for round $i+1$, we construct item arrivals as follows.
Without loss of generality, suppose $\mathcal{S}_{1-b} = \{t_1,
t_2, \ldots, t_l\}$. Let $\beta = \frac{\epsilon'(2\phi -
\epsilon')}{\phi - \epsilon'}$. We first generate
 $\beta m_i$ copies of $t_1$, and then $\beta m_i$
copies of $t_2$, \dots, then $\beta m_i$ copies of $t_l$, in sequence.
After these items we end round $i$ and start round
$i+1$. At this turning point, the total number of items is
$$m_{i+1} = m_i + l\cdot \beta m_i = \frac{\phi}{\phi-\epsilon'} m_i.$$
Now the frequency of each item in the set $\mathcal{S}_{1-b}$ is
$$(\phi - \epsilon')m_i + \beta m_i = \phi \cdot \frac{\phi}{\phi -
  \epsilon'} m_i = \phi m_{i+1},$$
and the frequency of each item in $\mathcal{S}_{b}$ remains the
same, that is, $\phi m_i = (\phi - \epsilon') m_{i+1}$.  Now we have
  restored the invariant and can start round $i+1$.

Finally, we bound the number of rounds.  Since the total number of
items $m_i$ increases by a $\phi/(\phi-\eps')$ factor in each
round, the total number of rounds is
$\Theta(\log_{\frac{\phi}{\phi-\epsilon'}}n)$. Consequently, the
total number of changes in the set of heavy hitters (from
non-heavy hitters to heavy hitters) is $l \cdot
\Theta(\log_{\frac{\phi}{\phi-\epsilon'}}n) =
\Omega(\frac{1}{\phi-\eps} \cdot \frac{\phi-\eps'}{\eps'} \log n)
= \Omega(\log n/ \epsilon)$.
\end{proof}

Now we go back to the distributed scenario and consider the cost of
communication for ``recognizing'' each change. Because we allow some
approximation when classifying heavy hitters and non-heavy hitters, the
valid time to report a change is actually a time interval, from
the time when its frequency just passes $(\phi-\eps)|A|$ to the time when
its frequency reaches $\phi|A|$.  As long as the tracking algorithm signals
the change within this interval, the algorithm is considered to be correct.
Consider the construction in the proof of
Lemma~\ref{lem:numberOfOutputUpdates}.  In round $i$, the transition
interval from a non-heavy hitter to a heavy hitter for an item $t$ must lie
inside the period in which the $\beta m_i$ copies of $t$ arrive.  Below we
will show that in order for the coordinator to signal the change within
this period, $\Omega(k)$ messages have to be exchanged in the worst case
using an adversary argument.

Before presenting the lower bound proof, let us be more precise about
the computation model.  Recall that in the introduction, the model
forbids a site to spontaneously initiate communication or change its
local status; actions can only be triggered as a result of the
arrival of an item at this site, or in response to the coordinator.  Note
that for deterministic algorithms this is not a restrictive
assumption. 
In our case, since we only care about the frequency of a particular
item $t$ increasing from $m_i$ to $m_i+\beta m_i$, we may assume that
each site $S_j$ has a {\em triggering threshold} $n_j$, meaning that
$S_j$ will only initiate communication when the number of copies of
$t$ received by $S_j$ is $n_j$.  When all the communication triggered
by the arrival of an item finishes, all the sites that have
participated are allowed to update their triggering thresholds, but
the rest of the sites must retain their old thresholds.  

\begin{lemma}
\label{lem:costEachUpdate} To correctly recognize a change in
the heavy hitters under the input constructed in the proof of
lemma~\ref{lem:numberOfOutputUpdates}, any deterministic algorithm has to
incur a communication cost of $\Omega(k)$.
\end{lemma}

\begin{proof}
  We will construct an adversary who will send the $\beta m_i$ copies of
  $t$ to the sites in a way such that at least $\Omega(k)$ sites must
  communicate with the coordinator.  Since we are dealing with
  deterministic algorithms, we may assume that the adversary knows the
  triggering thresholds $n_j$ at any time.

Initially, we must have
\begin{equation}
\label{eq:lb}\sum_{j=1}^{k} (n_j-1) <
  \beta m_i.
\end{equation}
Otherwise, the adversary can send $n_j-1$ copies to $S_j$ for all $j$
without triggering any communication, and make the algorithm miss the
change.  Therefore there must be some $j$ such that $n_j \le \beta m_i / k
+ 1 \le 2 \beta m_i/k$.  The adversary first sends $2\beta m_i/k$ copies of
$t$ to $S_j$.  $S_j$ will then communicate with the coordinator at least
once.  After the first $2\beta m_i/k$ copies, the new triggering thresholds
must still satisfy (\ref{eq:lb}).  Similarly, there is some $n_{j'} \le
2\beta m_i /k$, and the adversary will send another $2\beta m_i /k$ copies
of $t$ to $S_{j'}$.  Such a process can be repeated for $\frac{\beta
  m_i}{2\beta m_i/k} = \Omega(k)$ times, triggering at least $\Omega(k)$
messages of communication.
\end{proof}

The following lower bound follows immediately from
Lemma~\ref{lem:numberOfOutputUpdates} and
Lemma~\ref{lem:costEachUpdate}, for the reason that the tracking
algorithm has to correctly and continuously maintain the whole set
of heavy hitters.

\begin{theorem}
\label{thm:lowerbound-hh}
Any deterministic algorithm that continuously tracks the $\phi$-heavy
hitters has to incur a total communication cost of $\Omega(k/\eps\cdot \log
n)$, for any $\phi>3\eps$.
\end{theorem}

\paragraph{Remark.}
Note that our lower bound above is actually lower bound on the number of
messages required.  Also recall that our algorithm in
Section~\ref{sec:upper-bound} sends $O(k/\eps\cdot\log n)$ messages and
each message if of constant size.  Our lower bound implies that one cannot
hope to reduce the number of messages by making each of them longer.




\section{Tracking the Median}
\label{sec:track-median}

In this section we first present an algorithm to track any $\phi$-quantile
for $0 \le \phi \le 1$.  For ease of presentation we describe how to track
the median (the $1/2$-quantile); the generalization to any $\phi$-quantile
is straightforward. Then we give a matching lower bound.

\subsection{The upper bound}
For simplicity we assume that all the items in $A$ are distinct; issues
with ties can be easily resolved by standard techniques such as symbolic
perturbation.  We divide the whole tracking period into $O(\log n)$ rounds;
whenever $|A|$ doubles, we start a new round.  In the following we focus on
one round, and show that our median-tracking algorithm has a communication
cost of $O(k/\eps)$.

Let $m$ be the cardinality of $A$ at the beginning of a round.  Note that
$m$ is fixed throughout a round and we always have $m \le |A|$. 
The main
idea of our algorithm is to maintain a 
dynamic set of disjoint intervals in the coordinator (by
maintaining a set of separating items), such that each interval
contains between $\frac{\eps}{8}m$ and $\frac{\eps}{2}m$ items.
We first show that if we have such a set of intervals, the median
can be tracked efficiently.  Afterward we discuss how to maintain
these intervals.

Let $M$ denote the approximate median that is kept at the coordinator.  We
maintain two counters $C.\Delta(L)$ and $C.\Delta(R)$, counting the number
of items that have been received at all sites to the left and the right of
$M$, respectively. These two counters are maintained as underestimates with
an absolute error at most $\frac{\eps}{8} m$, by asking each site to send
in an update whenever it has received $\frac{\eps}{8k}m$ items to the left
or right of $M$.  So the cost of maintaining them is $O(k/\eps)$.

Whenever $|C.\Delta(L) - C.\Delta(R)| \ge \frac{\eps}{2} m$, we
update $M$ as follows.

\begin{enumerate}
\item  Compute $C.L$ and $C.R$ as the total number of items to the left and
  the right of $M$.  W.l.o.g., suppose $C.L > C.R$ and let $d = (C.L -
  C.R)/2$.

\item Compute a new median $M'$ such that $|r(M) - r(M') - d| \le
  \frac{\eps}{4} m$ where $r(M)$ is the rank of $M$ in $A$.  Update $M$ to
  $M'$.  Note that $M'$ is at most $\frac{\eps}{4}m$ items away from the
  exact median.  We will describe how to compute such an $M'$ shortly.

\item Reset $C.\Delta(L)$ and $C.\Delta(R)$ to 0.
\end{enumerate}

For the correctness of the algorithm, we can show that our
tracking algorithm always maintains an approximate median that is
at most $\frac{\eps}{4}m + \frac{3\eps}{4}m = \eps m$ items away
from the exact median.  The first term $\frac{\eps}{4}m$ is due to
the fact that whenever we update $M$, $M$ is within an error of
at most $\frac{\eps}{4}m$ to the exact median.  The second term
$\frac{3\eps}{4}m$ accounts for the error introduced by the
triggering condition $|C.\Delta(L) - C.\Delta(R)|$ monitored in
the coordinator. Note that we keep both $C.\Delta(L)$ and
$C.\Delta(R)$ within an additive error of at most $\frac{\eps}{8}
m$ and whenever $|C.\Delta(L) - C.\Delta(R)| \ge \frac{\eps}{2}
m$, we initiate an update.  Therefore, the total error
introduced is at most $2 \cdot \frac{\eps}{8}m + \frac{\eps}{2}
m =  \frac{3\eps}{4}m$.

Now we analyze the communication cost. Step 1 could be done by
exchanging $O(k)$ messages. For step 2, first note that $ d \le
\eps m$ since by the reasoning above, $M$ is still an
$\eps$-approximate median.  Next, we can find $M'$ quickly with
the help of the set of intervals. We start by finding the first
separating item $Y_1$ of the intervals to the left of $M$, and
then collect information from all sites to compute the number of
items in the interval $[Y_1, M]$, say $n_1$.  If $|n_1 - d| \le
\frac{\eps}{2} m$, we are done; otherwise we go on to pick the
second separating item $Y_2$ to the left of $M$, and check if
$|n_2 - d|\le \frac{\eps}{2}m$, where $n_2$ is the number of items
in the interval $[Y_2, M]$. It is easy to see that after at most
$O(1)$ such probes, we can find an item $Y$ such that the rank
difference between $Y$ and the exact median is no more than
$\frac{\eps}{2}m$. Note that the cost of each probe is $O(k)$ thus the
total cost of step 2 is $O(k)$. Finally, we update $M$ at most $O(1/\eps)$ times within
a single round, since each update increases $|A|$ by at least a
factor of $1 + \frac{\eps}{2}$. To sum up, the total cost of
the algorithm within a round is $O(k/\eps)$ provided that the
dynamic set of intervals are maintained.

\paragraph{Maintaining the set of intervals.}
When a new round starts, we initialize
the set of intervals as follows: Each site $S_j\ (1 \le j \le k)$
computes a set of intervals, each containing $\frac{\eps |A_j|}{32}$ items,
where $A_j$ stands for the set of items $S_j$ has received, and then sends
the set of intervals to the coordinator (by sending those separating
items).  Then the coordinator can compute the rank of any $x\in U$ with an
error of at most $\sum_{j=1}^k \frac{\eps}{32} |A_j| = \frac{\eps}{32} m$,
therefore it can compute a set of intervals, each of which contains at
least $\frac{\eps}{8}m$ and at most $\frac{\eps}{4} m$ items. After the
coordinator has built the set of intervals, it broadcasts them to all the
$k$ sites, and then computes the exact number of items in each interval.
  The cost of each rebuilding is $O({k}/{\eps})$.

During each round, each site $S_j$ maintains a counter for each interval as
new items arrive. And whenever the local counter of items in some interval
$I$ has increased by $\frac{\eps}{4k}m$, it sends a message to
the coordinator and the coordinator updates the count for interval $I$
accordingly.  Whenever the count of some interval in the coordinator $C$
reaches $\frac{\eps}{4} m$, the coordinator splits the interval into two
intervals, each of which containing at least $\frac{\eps}{8}m$ and at most
$\frac{\eps}{4}m$ items.  To perform such a split, we can again call the
rebuilding algorithm above, except that the rebuilding is only applied to
the interval $I$, so the cost is only $O(k)$.

The correctness of algorithm is obvious. The total communication
cost of interval splits is $O({k}/{\eps})$ in each round, since
there are at most $O(1/\eps)$ splits and each split incurs a
communication cost $O(k)$.

\begin{theorem}
  There is a deterministic algorithm that continuously tracks the
  $\eps$-approximate median (and generally, any $\phi$-quantile $(0 \le
  \phi \le 1)$) and incurs a total communication cost of $O(k/\eps\cdot
  \log n)$.
\end{theorem}

\paragraph{Implementing with small space.}
Similar to our heavy hitter tracking algorithm, instead of maintaining the
intervals exactly at each site, we can again deploy a sketch that maintains
the approximate $\eps'$-quantiles for some $\eps'=\Theta(\eps)$ to maintain
these intervals approximately.  Suppose we use the Greenwald-Khanna sketch
\cite{greenwald01:_space}, then we can implement our $\phi$-quantile
tracking algorithm with $O(1/\eps\cdot \log(\eps n))$ space per site and
amortized $O(\log n)$ time per item.

\subsection{The lower bound}

The idea of the proof of the lower bound is similar as that for
the heavy hitters. We try to construct a sequence of input with
the following properties.
\begin{enumerate}
\item The median will change at least $\Omega(\log n/\eps)$ times.

\item To correctly recognize each update, any deterministic
algorithm has to incur a communication cost of $\Omega(k)$.
\end{enumerate}

Consider the following construction. The universe consists of only
two items $0$ and $1$. We divide the whole tracking period to
several rounds and let $m_i$ be the number of items at the
beginning of round $i$. We maintain the following invariant: When
round $i$ starts, the frequency of item $b$ is $(0.5 -
2\eps)m_i$ and the frequency of item $1-b$ is $(0.5 +
2\eps)m_i$, where $b = i \bmod 2$. This could be done by inserting
$\frac{4\eps}{0.5 - 2\eps} m_i$ copies of $b$ during round $i$
and then start a new round. It is easy to see that there will be
at least $\Omega(\log n/\eps)$ rounds and the median will change
at least once during each round, therefore the total number of
changes of the median is $\Omega(\log n/\eps)$. For the second
property, we can invoke the same arguments as that for
Lemma~\ref{lem:costEachUpdate}.
Combining the two properties, we have the following.

\begin{theorem}
\label{thm:lowerbound-median}
  Any deterministic algorithm that continuously tracks the approximate
  median has to incur a total communication cost of
  $\Omega(k/\eps\cdot \log n)$.
\end{theorem}


\newcommand{\T}{\mathcal{T}}
\newcommand{\hmax}{h}

\section{Tracking All Quantiles}
\label{sec:tracking-quantiles}

In this section, we give a tracking algorithm so that the
coordinator $C$ always tracks the $\eps$-approximate
$\phi$-quantiles for all $0\le \phi \le 1$ simultaneously.  We will solve
the
following equivalent problem: The coordinator is required to
maintain a data structure from which we can extract the rank
$r(x)$ for any $x\in U$ in $A$ with an additive error at most
$\eps |A|$. We still assume that all items in $A$ are distinct.

We divide the whole tracking period into $O(\log n)$ rounds.  In each round
$|A|$ roughly doubles. We will show that the algorithm's cost in each round
is $O(k/\eps \cdot \log^2\frac{1}{\eps})$.  The algorithm restarts itself
at the beginning of each round, therefore the total communication of the
algorithm will be $O(k/\eps \cdot \log n \log^2\frac{1}{\eps})$.

\paragraph{The data structure.}
Let $m$ be the cardinality of $A$ at the beginning of a round. The data
structure is a binary tree $\T$ with $\Theta(1/\eps)$ leaves.  The root $r$
of $\T$ corresponds to the entire $A$.  It stores a splitting element $x_r$
which is an approximate median of $A$, i.e., it divides $A$ into two parts,
either of which contains at least $(\frac{1}{2}- \alpha)|A|$ and at most
$(\frac{1}{2}+\alpha)|A|$ items, for some constant $0<\alpha<\frac{1}{2}$.
Then we recursively build $r$'s left and right subtrees on these two parts
respectively, until there are no more than $\eps m /2$ items left.  It is
clear that $\T$ has $\Theta(1/\eps)$ nodes in total, and has height at most
$\hmax = \log_{\frac{1}{2}+\alpha} \frac{\eps}{2} =
\Theta(\log\frac{1}{\eps})$, though it is not necessarily balanced.  Each
node in $\T$ is naturally associated with an interval.  Let $I_u$ be the
interval associated with $u$.  Then $I_r$ is the entire $U$; suppose $v$
and $w$ are $u$'s children, then $I_u$ is divided into $I_v$ and $I_w$ by
$x_u$.  Set $\theta = \frac{\eps}{2\hmax}$.  Each node $u$ of $\T$ is in
addition associated with $s_u$, which is an underestimate of $|A\cap I_u|$
with an absolute error of at most $\theta m$, i.e., $|A\cap I_u| - \theta m
\le s_u \le |A\cap I_u|$. Please see Figure~\ref{fig:tracking_quan} for an
illustration of the data structure.

\begin{figure}
\begin{center}
\includegraphics[width=12cm]{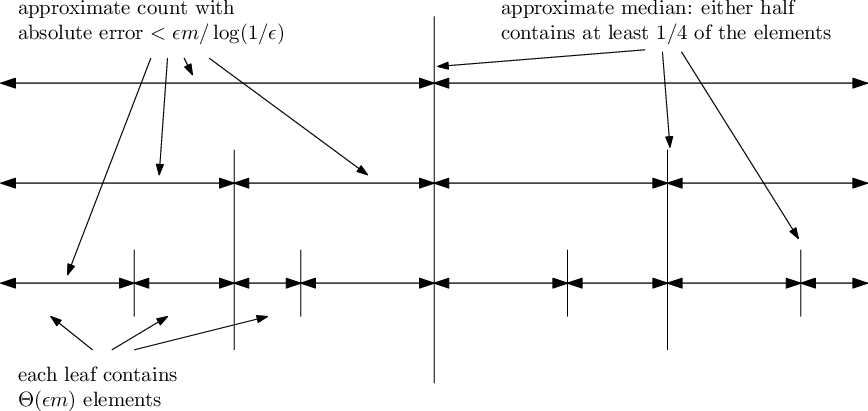}
\caption{The data structure that can be used to extract the rank of any
$x$ with absolute error $<\epsilon m$.} \label{fig:tracking_quan}
\end{center}
\end{figure}

If the coordinator has such a data structure, it is not difficult to see
that we can compute the rank of $x$ with an absolute error of at most $\eps
m$.  For a given $x$, we first search down the binary tree and locate the
leaf $v$ such that $x \in I_v$.  As we go along the root-to-leaf path,
whenever we follow a right child, we add up the $s_u$ of its left sibling.
In the end we add up $\hmax$ such partial sums, each contributing an error
of at most $\theta m$, totaling $\theta m \cdot \hmax = \eps m /2$.
Finally, since $|A \cap I_v| < \eps m /2$, the sum of all the $s_u$'s for
the preceding intervals of $x$ is off by at most $\eps m$ from the actual
rank of $x$.

\paragraph{Initialization.}
At the beginning of each round, we initialize the data structure
similarly as in Section~\ref{sec:track-median}. Suppose the set of
items at $S_j$ is $A_j$. Each site $S_j$ builds its own structure
$S_j.\T$, but with $\eps/32$ as the error parameter, and ships to
$C$.  This costs a communication of $O(k/\eps)$. Note that
$S_j.\T$ allows one to extract the rank of any $x$ within $A_j$
with an error of $\eps/32\cdot |A_j|$. By querying each $S_j.\T$, the
coordinator can compute the rank of any $x$ with an error of
$\sum_{i=1}^k \frac{\eps}{32} |A_i| = \frac{\eps}{32} m$, which is
enough for the coordinator to build its own $C.\T$.  In
particular, all the splitting elements can be chosen to be within
a distance of $\frac{\eps}{32} m$ to the real median. After
building $C.\T$, the coordinator broadcasts it to all the sites,
costing communication $O(k/\eps)$.  Now each site $S_j$ knows how
$U$ is subdivided into those $\Theta(1/\eps)$ intervals represented by the
binary tree $\T$. Then for
each interval $I_u$, it computes $|A_j \cap I_u|$ and sends the count to
$C$, so that the coordinator has all the exact partial sums $s_u$
to start with. It is easy to see that the total communication cost
for initializing the data structure is $O(k/\eps)$.

\paragraph{Maintaining the partial sums.}
As items arrive, each site $S_j$ monitors all the intervals $I_u$
in $\T$. For each $I_u$, every time the local count of items in
$I_u$ at $S_j$ has increased by $\theta m / k$, it sends an
updated local count to $C$.  Thus in the worst case, each site is
holding $(\theta m /k -1)$ items that have not been reported,
leading to a total error of at most $\theta m$.  The cost of these
messages can be bounded as follows.  When $S_j$ sends a new count
for some interval $I_u$, we charge the cost to the $\theta m /k$
new items that have arrived since the last message for $I_u$, $O(k
/ (\theta m))$ each.  Since each item contributes to the counts of
at most $\hmax$ intervals, it is charged $O(\hmax)$ times, so the
total cost charged to one item is $O(\frac{k\hmax}{\theta m})$.
There are a total of $O(m)$ items in a single round, so the
overall cost is $O(k\hmax/\theta) = O(k/\eps \cdot
\log^2\frac{1}{\eps})$.

\paragraph{Maintaining the splitting elements.}
The maintenance algorithm above ensures that all the $s_u$ are within the
desired error bound.  We still need to take care of all the splitting
elements, making sure that they do not deviate from the real medians too
much.  Specifically, when we build $\T$, for any $u$ with children $v$ and
$w$, we ensure that
\begin{equation}
\label{eq:2}
\frac{3}{8}|A \cap I_u| \le |A \cap I_v| \le
\frac{5}{8} |A \cap I_u|.
\end{equation}
This property can be easily established during initialization, since $|A
\cap I_u| > \frac{\eps}{2}m$ for any internal node $u$ of $\T$, and we can
estimate $|A\cap I_v|$ with an error of $\frac{\eps}{32}m$.  In the middle
of the round, we maintain the following condition:
\begin{equation}
\label{eq:3}
\frac{1}{4} s_u \le s_v \le \frac{3}{4} s_u.
\end{equation}
Recall that $s_u$ (resp.\ $s_v$) is an estimate of $|A \cap I_u|$ (resp.\
$|A \cap I_v|$) with an error of at most $\theta m$.  As long as
(\ref{eq:3}) holds, we have
\[
\frac{1}{4} (|A \cap I_u| - \theta m) \le \frac{1}{4} s_u \le s_v \le |A
\cap I_v| + \theta m.
\]
Rearranging,
\[
|A \cap I_v| \ge \frac{1}{4}|A\cap I_u| - \frac{5}{4}\cdot \frac{\eps}{2h}
m \ge \frac{1}{4}|A\cap I_u| - \frac{5}{4}\cdot \frac{1}{h}|A\cap I_u| \ge
\frac{3}{32} |A\cap I_u|,
\]
for $h\ge 8$. (Note that assuming $h$ larger than any constant does not
affect our asymptotic results.)  Similarly, we also have $|A\cap I_v | \le
\frac{29}{32} |A \cap I_u|$.  Thus condition~(\ref{eq:3}) ensures that the
height of $\T$ is bounded by $h = \Theta(\log\frac{1}{\eps})$.

Whenever (\ref{eq:3}) is violated, we do a partial rebuilding of the
subtree rooted at $u$ to restore this condition.  If multiple conditions
are violated at the same time, we rebuild at the highest such node.  To
rebuild the subtree rooted at $u$, we apply our initialization algorithm,
but only for the range $I_u$.  This incurs a cost of $O(k \frac{|A\cap
  I_u|}{\eps m})$, since we are essentially building a new data structure
on $|A \cap I_u|$ elements with error parameter $\eps' = \eps m / |A \cap
I_u|$.  After rebuilding, we have restored (\ref{eq:2}) for $u$ and all its
descendants.

It remains to bound the cost of the partial rebuildings.  Similarly as
before, we can show that when (\ref{eq:3}) is violated, we must have
\begin{equation}
  \label{eq:4}
|A \cap I_v| < \frac{21}{64}|A\cap I_u|,
\end{equation}
or
\begin{equation}
\label{eq:5}
   |A \cap I_v| > \frac{43}{64}|A\cap I_u|,
\end{equation}
assuming $h \ge 16$.  Note that both $|A\cap I_v|$ and $|A\cap I_u|$ may
increase.  From (\ref{eq:2}) to (\ref{eq:4}), $|A \cap I_u|$ must increase
by $\Omega(|A\cap I_v|) = \Omega(|A\cap I_u|)$; from (\ref{eq:2}) to
(\ref{eq:5}), $|A\cap I_v|$ must increase by $\Omega(|A \cap I_u|)$, which
implies that $|A\cap I_u|$ must also increase by $\Omega(|A \cap I_u|)$
since $I_v \subset I_u$.  This means that between two partial rebuildings
of $u$, $|A \cap I_u|$ must have increased by a constant factor.  Thus, we
can charge the rebuilding cost of $u$ to the $\Omega(|A \cap I_u|)$ new
items that have arrived since the last rebuilding,  $O(k/(\eps m))$ each.
Since each item is contained in the intervals of $O(h)$ nodes, it is
charged a cost of $O(hk/(\eps m))$ in total.  Therefore, the total cost of
all the partial rebuildings in this round is $O(hk/\eps) = O(k/\eps \cdot
\log\frac{1}{\eps})$.

\paragraph{Maintaining the leaves.}
Finally, we need to make sure that $|A\cap I_v| \le \frac{\eps}{2} m$ for
each leaf $v$ as required by the data structure.  During initialization, we
can easily ensure that $\frac{1}{8} \eps m \le |A\cap I_v| \le
\frac{3}{8}\eps m$.  During the round, the coordinator monitors $s_v$, and
will split $v$ by adding two new leaves below $v$ whenever $s_v >
(\frac{\eps}{2}-\theta)m$.  Since $s_v$ has error at most $\theta m$, this
splitting condition will ensure that $|A\cap I_v| \le \frac{\eps}{2} m$.
To split $v$, we again call our initialization algorithm on the interval
$I_v$, incurring a cost of $O(k\frac{|A\cap I_v|}{\eps m}) = O(k)$.  Since
we create at most $O(1/\eps)$ leaves in this entire round, the total cost
for all the splittings is $O(k/\eps)$.

\smallskip Putting everything together, we obtain the following result.

\begin{theorem}
  There is a deterministic algorithm that continuously tracks the
  $\phi$-quantiles for all $0\le \phi \le 1$ simultaneously and incurs a
  total communication cost of
  $O(k/\eps\cdot \log n \log^2\frac{1}{\eps})$.
\end{theorem}

\paragraph{Implementing with small space.}
Similar as before, instead of maintaining the counts in the intervals
associated with $\T$ exactly at each site, we can again deploy a sketch
that maintains the approximate $\eps'$-quantiles for some
$\eps'=\Theta(\theta)$ to maintain these intervals approximately.  Suppose
we use the Greenwald-Khanna sketch \cite{greenwald01:_space}, then we can
implement our all-quantile tracking algorithm with $O(1/\theta\cdot
\log(\theta n)) = O(1/\eps \cdot \log\frac{1}{\eps}\log(\eps n))$ space per
site and amortized $O(\log n)$ time per item.


\section{Open Problems}
We have restricted ourselves to deterministic algorithms in the paper.  If
randomization is allowed, simple random sampling can be used to achieve a
cost of $O((k+1/\eps^2)\cdot \polylog(n,k,1/\eps))$ for tracking both the
heavy hitters and the quantiles.  This observation has been well exploited
in maintaining the heavy hitters and quantiles for a single stream when
both insertions and deletions are present (see e.g.\
\cite{gilbert02:_how}).  This breaks the deterministic lower bound for
$\eps = \omega(1/k)$.  It is not known if randomization can still help for
smaller $\eps$.  Deriving lower bounds for randomized algorithms is also an
interesting open problem.  Another possible direction is to design
algorithms to track the heavy hitters and quantiles within a sliding window
in the distributed streaming model.

\bibliographystyle{abbrv}

\end{document}